\documentclass[11pt,article,onecolumn,draft]{IEEEtran}

\usepackage{cite}

\ifCLASSINFOpdf
\else
\fi
\usepackage{tikz} 
\usepackage{pgfplots}
\usepackage{tumcolor}
\usetikzlibrary{patterns,decorations}

\usepackage[cmex10]{amsmath}
\usepackage{amssymb}
\usepackage{fixmath}
\usepackage{amsbsy}
\usepackage{algorithmic}
\usepackage{algorithm}
\usepackage{amsthm}

\newtheorem{mythm}{Theorem}

\usepackage{array}

\ifCLASSOPTIONcompsoc
  \usepackage[caption=false,font=normalsize,labelfont=sf,textfont=sf]{subfig}
\else
  \usepackage[caption=false,font=footnotesize]{subfig}
\fi

\usepackage{fixltx2e}

\usepackage{dane}
\usepackage{todonotes}

\begin{document}

\title{Channel Estimation in Massive MIMO Systems}

\author{David Neumann, Michael Joham, and Wolfgang Utschick\\
	  Department for Electrical and Computer Engineering \\
	    Technische Universit\"at M\"unchen \\
	      80290 Munich, Germany\\\small
	        \{\texttt{d.neumann,joham,utschick}\}\texttt{@tum.de}%
\thanks{This work extends our previous considerations on semi-blind channel estimation in~\cite{neumann_suppression_2014}.}%
}

\markboth{}%
{}

% make the title area
\maketitle

\begin{abstract}
	We introduce novel blind and semi-blind channel estimation methods
	for cellular time-division duplexing systems with a large number of
	antennas at each base station.  The methods are based on the
	maximum a-posteriori principle given a prior for the distribution
	of the channel vectors and the received signals from the uplink
	training and data phases.  Contrary to the state-of-the-art massive
	MIMO channel estimators which either perform linear estimation
	based on the pilot symbols or rely on a blind principle, the
	proposed semi-blind method efficiently suppresses most of the
	interference caused by pilot-contamination.  The simulative
	analysis illustrates that the semi-blind estimator outperforms
	state-of-the-art linear and non-linear approaches to the massive
	MIMO channel estimation problem.
\end{abstract}
%
%\begin{IEEEkeywords}
%\end{IEEEkeywords}

\IEEEpeerreviewmaketitle

\section{Introduction}

Deployment of large numbers of antennas at the base stations in a cellular
network can lead to huge gains in both, spectral and energy efficiency. The
increase in energy efficiency is mostly due to the array gain and the high
spectral efficiency is achieved by serving several users simultaneously
through spatial multiplexing.

Under favorable propagation conditions~\cite{rusek_scaling_2013}, namely,
i.i.d. channel coefficients for the links from the different antennas at a
base station to a user terminal, the spatial channels to distinct users are
close to orthogonal in such a massive MIMO system due to the law of large
numbers. Consequently, the interference caused by spatial multiplexing can
be suppressed with simple signal processing methods, e.g., matched
filters~\cite{marzetta_noncooperative_2010,hoydis_massive_2013,rusek_scaling_2013}
or even constant envelope precoding~\cite{mohammed_per-antenna_2013}.

Even these simple beamforming techniques require accurate channel state
information (CSI) at the base station to enable the interference
suppression capability.  A closed-loop setup is necessary for channel
estimation in frequency division duplex (FDD) systems, where the downlink
channels must be estimated in a downlink training phase and fed back to the
base station to allow downlink beamforming.  Note that the worst-case
number of orthogonal pilots necessary to enable the estimation of the
vector channels by the single-antenna receivers is equal to the number of
base station antennas. Given a block fading model, where the channels are
constant in a coherence interval with fixed number of channel accesses, it
is thus challenging to implement an accurate closed-loop channel estimation
for a massive MIMO system, due to the large number of base station
antennas~\cite{marzetta_noncooperative_2010,rusek_scaling_2013}.
Additionally, the fed back CSI is deteriorated not only by the estimation
errors, but might also suffer from the quantization error due to the limited
rate feedback and outdating due to the delay between estimation and
application at the base stations.

Instead, we focus on time-division duplexing (TDD) massive MIMO systems and
exploit the channel reciprocity to acquire CSI in a timely manner. That is,
in each coherence interval, we have an uplink training phase where each
user transmits a pilot sequence. Thus, in a TDD system, the necessary
number of orthogonal training sequences is equal to the number of users. In
other words, the number of orthogonal pilot sequences is dramatically
smaller for TDD than for FDD massive MIMO systems, since the number of
users is much smaller than the number of antennas at the base station. 

For the downlink phase in TDD systems, the base station generates beamforming
vectors based on the channel estimates obtained during the previous uplink phase.
Note that such a usage of the uplink CSI for downlink beamforming is
impossible for FDD systems contrary to TDD systems. However, a
sophisticated calibration of the uplink and downlink signal chains is
necessary to enable the exploitation of the channel reciprocity in TDD
systems. In the following, we will assume perfect calibration implying that
reciprocity between uplink and downlink channels holds.

As the total number of orthogonal pilot sequences is limited by the channel
coherence interval and the requirement to actually transmit payload data, the pilot sequences must be reused in neighboring
cells. Therefore, the channel estimates contain interference from the
channels to users in neighboring cells with the same pilot sequences since
only the pilot sequences within one cell are guaranteed to be orthogonal.
The interference between the channels of the users during the uplink
training leads to interference in both the uplink and downlink data transmission phases in a reciprocity
based system. This effect is called pilot-contamination, which, due to the
resulting interference during the data transmission, poses a limit
on the achievable rate in a massive MIMO system that relies on linear
signal processing~\cite{marzetta_noncooperative_2010,
hoydis_massive_2013,gopalakrishnan_analysis_2011}. 

Several approaches were proposed to tackle the pilot-contamination problem.
With coordination between base stations, the impact of pilot-contamination
can be reduced, e.g., by appropriately scheduling the downlink and training
phases for neighboring cells~\cite{fernandes_inter-cell_2013} or by
coordinating the allocation of pilot-sequences in neighbouring
cells~\cite{neumann_pilot_2014,yin_coordinated_2013,huh_achieving_2012}.
Optimization of the trade-off between channel accesses employed for pilot
transmission and for data transmission further improves the total
throughput and fairness of the
system~\cite{hoydis_optimal_2011,neumann_amount_2014}.

Cooperative transmission from several base stations to a single user, so
called coordinated multi-point (CoMP) or network MIMO, was proposed for
multi-cell MIMO systems to enable inter-cell interference
management~\cite{gesbert_multi-cell_2010}. To this end, instantaneous CSI
is necessary at the central hub. In massive MIMO systems, however,
cooperative transmission strategies, so called pilot contamination
precoding (PCP), can be designed based on channel statistic
\cite{ashikhmin_pilot_2012,neumann_cdi_2015,lu_overview_2014,li_pilot_2013}.
In the asymptotic limit, when the number of base station antennas tends to
infinity, it is possible to apply zero-forcing techniques to remove all
interference using only statistical information about the
channel~\cite{ashikhmin_pilot_2012,neumann_cdi_2015}.

Independently of improved resource allocation or PCP, the interference and
noise in the channel estimates can alternatively be reduced by more
sophisticated channel estimation methods. The structure in the channel
covariance matrices can be exploited by linear minimum mean squared error
(MMSE) estimation~\cite{yin_coordinated_2013} or by heuristics based on
projecting onto principal and minor subspaces of the channel covariance
matrices~\cite{yin_dealing_2014,liu_hierarchical_2014}. For sparse
channels, other approaches involving results from compressed sensing theory
have been proposed~\cite{masood_efficient_2015,rao_distributed_2014}.
Approaches that exploit the structure of the channel might reduce the
necessary amount of pilot sequences in FDD massive MIMO
systems~\cite{rao_distributed_2014,choi_downlink_2014,nam_joint_2014}. As
pilot-contamination is also a sever problem in the uplink, similar
techniques have been proposed for equalization~\cite{adhikary_uplink_2014}.

Previous work on the suppression of pilot-contamination during the channel
estimation by the application of non-linear signal processing has been
reported in~\cite{mueller_blind_2013}. The blind channel estimation
approach proposed in~\cite{mueller_blind_2013} is based on the idealistic
assumption, that interfering users from neighboring cells always have
weaker channels than the desired users. Under this assumption, it is
possible to separate the desired uplink signal space and the interference
space in the asymptotic limit of an infinite number of users and antennas,
for a fixed ratio of number of users and antennas, just by applying a
principal component analysis to the received uplink data.  Another related
blind channel estimation method was proposed in~\cite{ngo_evd-based_2012},
where the principal component analysis of the uplink data together with
prior information on the quality of the channels to the different users is
used to estimate the channel vectors up to a scalar ambiguity.
Alternatively or additionally to the principal component analysis of the
uplink data, also information from the decoder can be used to improve the
channel estimation~\cite{ma_data-aided_2014,ngo_evd-based_2012}. 

A semi-blind channel estimation approach for the interference-free case was
proposed in~\cite{jagannatham_whitening-rotation-based_2006}. The
approach is based on the assumption that the blind estimate of the channel is
unique up to a unitary rotation. The unitary ambiguity is then estimated
with the help of pilot signals. In~\cite{neumann_suppression_2014}, we
proposed a heuristic semi-blind method based on the maximum-a-posteriori
(MAP) problem formulation, which, however, is not sufficiently robust in
the case of inaccurate subspace information from the principal component
analysis of the uplink data.

We take a methodical approach to the uplink channel estimation problem in 
massive MIMO systems with the following contributions.
\begin{itemize}
	\item We review results on training-based estimation in
		Section~\ref{sec:training-based} and give the analytical
		solution for the maximum-a-posteriori (MAP) estimate based
		on the received training data, which, for the considered
		Gaussian system model, is equivalent to the training-based
		minimum mean squared error (MMSE) estimate.

	\item In Section~\ref{sec:blind}, we derive an analytical solution of the
		MAP estimator for the blind case, i.e., the channel estimation only
		relies on the received data but no training data is used. The
		resulting estimate is based on the singular value decomposition of
		the received uplink data and is similar in structure to the one
		proposed in~\cite{ngo_evd-based_2012}.

	\item In Section~\ref{sec:semi-blind}, we formulate the MAP estimator for
		the combined data and training based estimation, i.e., semi-blind
		channel estimation. However, the resulting optimization problem
		turns out to be non-convex. Therefore, a gradient ascent method is
		briefly described to obtain a local optimum.

	\item We discuss heuristic suboptimal methods for semi-blind estimation
		which can be used as a starting point for a gradient based search
		on the original MAP problem.

	\item The numerical results, presented in Section~\ref{sec:results} to compare the different
		estimation approaches, demonstrate the performance gains
		using the proposed semi-blind estimation methods.
\end{itemize}

\subsection{Notation}

Throughout this paper, $(\cdot)\he$ denotes the conjugate transpose of a
matrix and $(\cdot)^+$ denotes the pseudo inverse. The notation $f(x) \propto
g(x)$ is short for $g(x) = c f(x)$ where $c$ is a constant which is
independent of $x$. With $\zeros$ and $\ones$ we denote vectors of
all-zeros and all-ones and $\id$ denotes the identity matrix. The operation
$[\cdot]_+$ is short for $\max(\cdot,0)$.

\section{System Model}

We consider a cellular network with $L$ cells. Without loss of generality,
we focus on the channel estimation at the base station in cell $1$. Let
\begin{equation}
	\vect H_i=[\vect h_{i1},\ldots,\vect h_{iK}] \in \mathbb C^{M\times K}
\end{equation}
denote the matrix of channel vectors corresponding to the users $k=1,\ldots,K$ in cell
$i$ to base station $1$, where $M$ is the number of antennas at the base
station in cell $1$ and $K$ is the number of users per cell, which, for the
sake of notational brevity, we assume is the same for all cells.

The channel vector from user $k$ in cell $i$ to base station $1$ is modeled as~\cite{marzetta_noncooperative_2010}
\begin{equation}\label{eq:chanmodsingle}
	\vect h_{ik} = \vect a_{ik} \sqrt{\beta_{ik}}
\end{equation}
where the entries of $\vect a_{ik}$ are i.i.d. complex Gaussian with zero
mean and unit variance, and the scaling factor $\beta_{ik}$ describes the
quasi-static shadow fading and the path-loss. Channel vectors to different
users are assumed to be independent. Consequently, we have
\begin{equation}\label{eq:chanmod}
	\H_i = \A_i \B_i^{1/2}
\end{equation}
with $\A_i = [\a_{i1},\ldots,\a_{iK}]$ and $\B_i = \diag(\beta_{i1},\ldots,\beta_{iK})$.

We consider a block fading model, i.e., the channel is constant in a
certain coherence interval~$T$ and is independent of the channel in the
next coherence interval. Of this coherence interval, $T_\text{ul}$ time
slots are used for uplink data transmission, $T_\text{tr} \geq K$ time
slots for the transmission of uplink pilot sequences, and the remaining time
slots are used for downlink data transmission. We assume that all terminals
use the same transmit power. The received uplink data is given
by
\begin{equation}
	\Yul = \sqrt{\rho_\text{ul}} \sum_{i=1}^L \H_i \vect
	X_i\he + \vect N_\text{ul} \in \mathbb C^{M \times T_\text{ul}}
\end{equation}
where $\rho_\text{ul}$ denotes the uplink signal-to-noise-ratio (SNR),
i.e., the uplink transmit power divided by the noise power at one
receive antenna, $\vect X_i \in \mathbb C^{ T_\text{ul} \times K}$ is the
uplink data transmitted by the users in cell $i$, and $\vect N_\text{ul}$ is the
additive noise. The entries in $\vect N_\text{ul}$ and $\vect X_i$ are
assumed to be i.i.d. complex Gaussian with zero mean and unit variance. 

For notational convenience, we collect all variables of the interfering
cells in $\tilde\H = [\H_2,\ldots,\H_L]$ and $\tilde\X =
[\X_2,\ldots,\X_L]$ leading to
\begin{equation}\label{eq:sys_uplink}
	\Yul = \sqrt{\rho_\text{ul}} \H_1 \vect X_1\he + \sqrt{\rho_\text{ul}} \tilde\H \tilde\X\he + \vect N_\text{ul} \in \mathbb C^{M \times T_\text{ul}}.
\end{equation}

Similarly, the received training data can be written as 
\begin{equation}\label{eq:sys_training}
	\Ytr = \sqrt{\ptr} \H_1 \vect \vPsi_1\he + \sqrt{\ptr} \tilde\H \tilde\vPsi\he + \vect N_\text{tr} \in \mathbb C^{M \times T_\text{tr}}
\end{equation} 
where the columns of $\vPsi_1 \in \mathbb C^{\Ttr \times K}$ are the, not
necessarily orthogonal, unit-norm training sequences used by the users in
cell $1$ and the columns of $\tilde\vPsi = [\vPsi_2,\ldots,\vPsi_L]\in
\mathbb C^{\Ttr \times (L-1) K}$ are the training sequences used in the
interfering cells. It is necessary to set $\ptr=\pul \Ttr$ for a
resulting average uplink SNR of $\pul$ since the training sequences are
unit-norm.

\section{Training Based Estimation}
\label{sec:training-based}

In this section, we review some well established results on training based
estimation that relate to the novel approaches presented in the following
sections. We will also use these training based estimates as a baseline in
the evaluation of the different estimation approaches in
Section~\ref{sec:results}.

One straightforward way to find a channel estimate based on the training
data without employing any prior information is the least squares (LS)
estimate
\begin{equation}\label{eq:ls}
	\H^\text{LS}_1 = \argmin_{\H_1} \left\lVert \frac{1}{\sqrt{\ptr}}\Ytr - \H_1 \vPsi_1\he \right\rVert\fro^2
\end{equation}
where $\left\lVert \cdot \right\rVert\fro$ denotes the Frobenius norm. Note
that the optimizer of~\eqref{eq:ls} is the maximum likelihood (ML) estimate
for the interference-free case due to the assumption that the noise $\vN_\text{tr}$
is complex Gaussian with i.i.d. zero-mean, unit-variance entries. The
solution to~\eqref{eq:ls} is simply given by
\begin{equation}\label{eq:ls_sol}
	\H^\text{LS}_1 = \frac{1}{\sqrt{\ptr}}\Ytr (\vPsi_1\he)^+.
\end{equation}

Consequently, for orthonormal pilot sequences, where
$(\vPsi_1\he)^+=\vPsi_1$, we simply correlate the received signals at all
antennas with each of the pilot sequences. Note that the LS estimate does not give
satisfactory results if the interference from the other cells during the
pilot phase cannot be neglected.

We also formulate the maximum likelihood (ML) problem for the joint
estimation of both, the desired and the interfering channels, as
\begin{align}
	(\H^\text{ML}_1, \tilde \H^\text{ML}) &= \argmax_{(\H_1,\tilde \H)} f_{\Ytr \vert \H_1,\tilde \H}(\Ytr \vert \H_1,\tilde \H) \notag \\
	&= \argmax_{(\H_1,\tilde \H)} l_\text{tr}(\H_1,\tilde \H) 
\label{eq:ml}
\end{align}
where
\begin{equation}
	f_{\Ytr \vert \H_1, \tilde \H}( \Ytr \vert \H_1,\tilde \H) \propto 
	\exp\left[ -\bigg\lVert \Ytr - \sqrt{\ptr}\sum_{j=1}^L \H_j\vPsi_j\he \bigg\rVert\fro \right]
\end{equation}
for the assumption of white Gaussian noise.

Therefore, the log-likelihood function reads as
\begin{equation}\label{eq:lh_tr}
	l_\text{tr}(\H_1,\tilde \H) = -\left\lVert \Ytr - \sqrt{\ptr}\H_1 \vPsi_1\he - \sqrt{\ptr}\tilde \H \tilde \vPsi\he \right\rVert\fro^2.
\end{equation}
Setting the derivative w.r.t. $\H_1^*$ and $\tilde \H^*$ to zero, we get the normal equations
\begin{equation}\label{eq:normeq}
	\bmat \H_1, & \tilde \H \emat \bmat \vPsi_1\he \\ \tilde \vPsi\he \emat \bmat \vPsi_1, & \tilde \vPsi \emat 
	= \frac{\Ytr}{\sqrt{\ptr}}\bmat \vPsi_1, & \tilde \vPsi \emat.
\end{equation}
Since $[\vPsi_1,\tilde \vPsi] \in \mathbb C^{\Ttr \times LK}$, the normal
equations~\eqref{eq:normeq} do not have a unique solution for $\H_1$ and
$\tilde \H$ if the number of pilot symbols $T_\text{tr}$ is smaller than
the total number of users $LK$ in the network. Note that if we set $\tilde \H =
\zeros$ in~\eqref{eq:ml}, i.e., if we neglect the channels to the interfering
users, we end up with the least squares problem~\eqref{eq:ls}.

If the priors are known, we can formulate the training based maximum
a-posteriori (MAP) estimator
\begin{align}
	(\H_1^\text{TR},\tilde \H^\text{TR}) &= \argmax_{(\H_1, \tilde \H)} f_{\Ytr \vert \H_1,\tilde \H}(\Ytr \vert \H_1,\tilde \H) f_{\H_1}(\H_1) f_{\tilde \H}(\tilde \H)\notag\\
		&= \argmax_{(\H_1,\tilde \H)} l_\text{tr}(\H_1,\tilde \H) + l_\text{pr}(\H_1,\tilde \H)
\end{align}
with
\begin{equation}\label{eq:lh_pr}
	l_\text{pr}(\H_1,\tilde \H)= -\tr(\H_1\B_1\inv\H_1\he) - \tr(\tilde \H \tilde \B\inv \tilde \H \he)
\end{equation}
due to~\eqref{eq:chanmod} where $\tilde \B = \blkdiag(\B_2,\ldots,\B_L)$,
and $l_\text{tr}(\vH_1, \tilde \vH)$ can be found in~\eqref{eq:lh_tr}.

The linear system of equations for the optimizer reads as
\begin{equation}\label{eq:trmapsys}
	\bmat \H_1, & \tilde \H \emat 
	\bmat \vPsi_1\he\vPsi_1 + \vB_1\inv & \vPsi_1\he \tilde\vPsi \\ \tilde \vPsi\he \vPsi_1 & \tilde\vPsi\he \tilde\vPsi + \tilde \B\inv \emat 
	= \frac{\Ytr}{\sqrt{\ptr}}\bmat \vPsi_1, & \tilde \vPsi \emat
\end{equation}
which is obtained by setting the derivatives w.r.t. $\vH_1^*$ and $\tilde
\vH^*$ to zero and has a unique solution due to the regularization by
$\B_1\inv$ and $\tilde \B \inv$. We explicitly solve~\eqref{eq:trmapsys} for $\H_1$ leading to
%\begin{multline*}
%	\H_1(\vPsi_1\he(\id - \tilde \vPsi(\tilde\vPsi\he\tilde\vPsi + \tilde\B\inv)\inv\tilde\vPsi\he)\vPsi_1 + \B_1\inv)\\ 
%	= \frac{\Ytr}{\sqrt{\ptr}}(\id- \tilde\vPsi(\tilde\vPsi\he\tilde\vPsi + \tilde\B\inv)\inv \tilde\vPsi\he)\vPsi_1.
%\end{multline*}
%Application of the matrix inversion lemma gives
%\begin{multline*}
%	\H_1(\vPsi_1\he(\id + \tilde\vPsi\tilde\B\tilde\vPsi\he)\inv\vPsi_1 + \B_1\inv) \\
%	= \frac{\Ytr}{\sqrt{\ptr}} (\id + \tilde\vPsi\tilde\B\tilde\vPsi\he)\inv \vPsi_1.
%\end{multline*}
%Substituting the definitions of $\tilde \vPsi$ and $\tilde \B$ yields
%\begin{multline}\label{eq:sys_reform}
%	\H_1^\text{TR} = \frac{\Ytr}{\sqrt{\ptr}} \Big( \id + \sum_{i=2}^L\vPsi_i\B_i\vPsi_i\he \Big)\inv \vPsi_1 \times \\
%	\Big( \vPsi_1\he(\id + \sum_{i=2}^L\vPsi_i\B_i\vPsi_i\he)\inv\vPsi_1 + \B_1\inv \Big)\inv.
%\end{multline}
\begin{align*}
	\H_1(\vPsi_1\he(\id - \tilde \vPsi(\tilde\vPsi\he\tilde\vPsi + \tilde\B\inv)\inv\tilde\vPsi\he)\vPsi_1 + \B_1\inv)
	= \frac{\Ytr}{\sqrt{\ptr}}(\id- \tilde\vPsi(\tilde\vPsi\he\tilde\vPsi + \tilde\B\inv)\inv \tilde\vPsi\he)\vPsi_1.
\end{align*}
Application of the matrix inversion lemma gives
\begin{align*}
	\H_1(\vPsi_1\he(\id + \tilde\vPsi\tilde\B\tilde\vPsi\he)\inv\vPsi_1 + \B_1\inv)
	= \frac{\Ytr}{\sqrt{\ptr}} (\id + \tilde\vPsi\tilde\B\tilde\vPsi\he)\inv \vPsi_1.
\end{align*}
Substituting the definitions of $\tilde \vPsi$ and $\tilde \B$ yields
\begin{align}\label{eq:sys_reform}
	\H_1^\text{TR} = \frac{\Ytr}{\sqrt{\ptr}} \Big( \id + \sum_{i=2}^L\vPsi_i\B_i\vPsi_i\he \Big)\inv \vPsi_1 \times
	\Big( \vPsi_1\he(\id + \sum_{i=2}^L\vPsi_i\B_i\vPsi_i\he)\inv\vPsi_1 + \B_1\inv \Big)\inv.
\end{align}
For our model, the MAP estimate is also the minimum mean squared error (MMSE) estimate,
because $f_{\Ytr,\H_1,\tilde\H}(\Ytr,\H_1,\tilde \H)$ is jointly Gaussian in $\H_1,\tilde\H$, and
$\Ytr$.

If we use the same orthonormal training sequences $\vPsi_i = \vPsi_1$ in
all cells $i$,~\eqref{eq:sys_reform} simplifies to 
\begin{equation}
	\H_1^\text{TR} = \frac{1}{\sqrt{\ptr}}\Ytr \vPsi_1 \B_1 \bigg( \id + \sum_{i=1}^L \B_i \bigg)\inv.
\end{equation}
In other words, the MMSE estimate is simply a scaled LS estimate~[cf.~\eqref{eq:ls_sol}] in this case and
consequently, the performance for linear matched filter or zero-forcing
precoding based on these estimation methods is equivalent. Note that this statement is
not true for the more general case of correlated channel coefficients, i.e.,
{$\h_{ik} \sim \CN(\zeros,\vR_{ik})$} where $\vR_{ik}$ is \emph{not} a
scaled identity matrix. In this case, the MMSE estimation can significantly
outperform the LS estimation, depending on the structure of the channel
covariance matrices.

In practical systems, the number of pilot symbols is significantly smaller
than the total number of users in the network. Thus, the training-based
estimates are subject to pilot contamination because of the necessary reuse
of pilot sequences. With $T_\text{tr} = K$ and the same set of orthonormal
pilot sequences reused in each cell, the least squares estimate is given by
\begin{equation}
	\H^\text{LS}_1 = \H_1 + \sum_{j=2}^L \H_j + \frac{1}{\ptr}\vN_\text{tr}\vPsi_1.
\end{equation}
The effect of the interference during the uplink channel estimation is
clearly visible. Furthermore, note that $\vN_\text{tr}\vPsi_1 \sim
\vN_\text{tr}$ due to the assumption of orthonormal pilots.

\section{Blind Estimation}
\label{sec:blind}

Previous work on blind channel and subspace estimation for massive MIMO
networks is based on the asymptotic orthogonality of both, channel vectors
and data symbol sequences for a large number of antennas $M$ and a large
number of received uplink data signals
$\Tul$~\cite{mueller_blind_2013,ngo_evd-based_2012}.

We refine the algorithm of~\cite{ngo_evd-based_2012} for the multi-cell
case, under the assumption that the slow fading coefficients
$\beta_{ik}$ to all users can be learned at the base stations and are therefore known~[see~e.g.,\cite{ashikhmin_pilot_2012}].

For blind channel estimation, the notational differentiation between desired and
interfering channels is not necessary. Hence, we collect all channels in a
single matrix 
\begin{align}\label{eq:channel}
	\H = [\H_1, \ldots, \H_L] = \vA\vB
\end{align}
with the corresponding
matrix of slow fading coefficients $\vB = \diag(\vB_1,\ldots,\vB_L)$ and
the matrix of fast fading coefficients $\vA = [\vA_1,\ldots,\vA_L]$ with i.i.d. complex Gaussian, unit-variance,
zero-mean entries. In this section, we will derive the analytical result for
the MAP estimator in the blind case, that is,
\begin{equation}\label{eq:MAPblindorig}
	\H^\text{BL} = \argmax_{\H} f_{\Yul \vert \H}(\Yul \vert \H) f_{\H}(\H).
\end{equation}
Note that the formulation in~\eqref{eq:MAPblindorig} is only based on the
received data $\vY$~[see~\eqref{eq:sys_uplink}] but is independent of the
received pilot signal $\vPhi$~[see~\eqref{eq:sys_training}]. Thus, a
quasi-closed-form solution can be obtained.

\begin{mythm} 
For $M > KL$, the blind MAP channel estimate is unique up to an unknown
complex phase for each channel vector, and the SVD of one possible
optimizer 
\begin{align}
	\H^\text{BL} = \vW \vXi \vT
\end{align}
can be calculated from the SVD of the uplink data $\Yul = \vU \vSigm \vV\he$ as follows.

The matrix of left singular vectors $\vW=\vU_{1:KL}$ is a matrix with
columns $1$ to~$KL$ of $\vU$. 

The matrix of right singular vectors $\vT = \vPi \in \{0,1\}^{KL \times KL}$
is a permutation matrix, such that the entries along the diagonal of
$\vPi\tp \vB\inv \vPi$ are sorted ascendingly. 

The singular values in the diagonal~$\vXi$ are given by
\begin{equation}
	\xi_n = \sqrt{\left[-\frac{1}{\pul} - \frac{\beta_{\pi\inv(n)}\Tul}{2} + \sqrt{ \frac{\beta_{\pi\inv(n)}^2\Tul^2}{4} + \frac{\beta_{\pi\inv(n)}\sigma_n^2 }{\pul}} \right]_+}.
\end{equation}

\end{mythm}

\begin{proof}
	The proof can be found in Appendix~\ref{sec:proof}
\end{proof}

The singular values are not necessarily relevant for practical purposes,
i.e., linear precoding and filtering. The main observation is here, that we
can estimate the channels up to a scalar ambiguity by applying an SVD to
the uplink data $\Yul$ and correctly assigning the left singular vectors to
the different users. The assignment of the channel vectors to users is
possible if the slow fading coefficients in $\vB$ are known.

Note that the estimated channel vectors are already orthogonal since they
are left singular vectors of $\vY$. Consequently, the matched filter and
the zero-forcing filter based on this estimate are equivalent.

A major cause of estimation errors of the blind estimator are users with
similar slow fading coefficients. Even for perfectly orthogonal channel
vectors, which is fulfilled for $M\rightarrow \infty$, the corresponding
subspaces can not be separated with the SVD of the uplink data and the
estimated channel vectors are a linear combination of the actual channel
vectors of those users with similar slow fading coefficients. To accurately
separate channels to different users with blind estimation we need not only
a large number of antennas, but also a large number of data samples within
one coherence interval. The first condition is met in a massive MIMO
system, but the length of the coherence interval is an inherent property of
the wireless channel and thus a sufficient number of data samples cannot be
guaranteed.

\section{Semi-blind Estimation}
\label{sec:semi-blind}

If only a small number of uplink data samples is available, the blind
estimation approach discussed in the previous section can deliver poor
results, because the subspace estimate might be inaccurate. In this
section, we therefore consider the joint semi-blind estimation based on
both, uplink data and training signals.

We first discuss the MAP problem for this setup, which is non-convex in the
semi-blind case.  Then, we introduce a heuristic estimation method based on
subspace projection that can be used to initialize an unconstrained solver
for the MAP problem and exhibits low-complexity.

We want to calculate an estimate of the channel $\vect H_1$ given the observations $\Yul$ of the data
and $\Ytr$ of the pilot signals. However, difficult marginalization steps can be avoided
by jointly estimating all channels. With $\vect H = [\vect H_1, \ldots, \vect
H_L]$~[see~\eqref{eq:channel}], the semi-blind MAP estimate for all channels to base station $1$ is given by
\begin{align}
	\vect H^\text{SB} &= \argmax_{\H} f_{\H \vert \Yul,\Ytr}(\H \vert \Yul,\Ytr) \notag \\ 
	&= \argmax_{\vect H} 
	f_{\Yul \vert \vect H}( \Yul \vert \vect H )
	f_{\Ytr \vert \vect H}( \Ytr \vert \vect H) 
	f_{\vect H}(\vect H). \notag\\
	&= \argmax_{\vect H} l_\text{tr}(\vect H) + l_\text{pr}(\vect H) + l_\text{ul}(\vect H)
	\label{eq:MAP}
\end{align}
since the received data signal $\vY$ and the received pilots $\vPhi$ are
independent when conditioned on the total channel matrix $\vH$.  All of the
probability density functions in~\eqref{eq:MAP} are circularly symmetric
complex Gaussian and follow directly from the system model
in~\eqref{eq:chanmod}, \eqref{eq:sys_uplink}, and~\eqref{eq:sys_training}.
The corresponding log-likelihood function for the training signal is given by~[cf.~\eqref{eq:lh_tr}]
\begin{align}\label{eq:lh_tr_semiblind}
	l_\text{tr}(\H) = -\left\lVert \Ytr - \sqrt{\ptr}\H \vPsi\he \right\rVert\fro^2.
\end{align}
with $\vPsi = [\vPsi_1,\ldots,\vPsi_L]$. The log-likelihood functions for
the prior and the uplink signals have been derived in the previous
sections, see  \eqref{eq:lh_pr_blind}, and~\eqref{eq:lh_ul}.

Due to the non-convex nature of the objective function in~\eqref{eq:MAP},
finding the global optimizer is difficult in general. However, since the
objective is differentiable, we can use any gradient based method to find a
local optimum from an initial guess.

For the derivatives, we have
\begin{equation}
	\frac{\partial l_\text{ul}(\vH)}{\partial \vect H^*} = (\id - \vect G \vect H^{\He})\Yul \Yul\he \vect G - \Tul \vect G
\end{equation}
for the data part, with
\begin{equation}
	\vect G = \vect H \bigg(\frac{1}{\pul}\id + \vect H^{\He} \vect H \bigg)^{-1}.
\end{equation}
For the prior part, we get
\begin{equation}
	\frac{\partial l_\text{pr}(\vH)}{\partial \vect H^*} = \vH \vB\inv = -[\vect H_1 \vect B_1\inv, \ldots \vect H_L \vect B_L\inv]
\end{equation}
and for the training part,
\begin{equation}
	\frac{\partial l_\text{tr}(\vH)}{\partial \vect H^*} =  \sqrt{\ptr}(\Ytr - \sqrt{\ptr}\H\vPsi\he )\vPsi.
\end{equation}
A suitable method for finding a local optimum for the relatively large
scale semi-blind MAP problem is, e.g., the
L-BFGS~\cite{nocedal_updating_1980} algorithm, which is a limited-memory
quasi newton method.

An accurate initial guess is essential to obtain a good performance. To this
end, we propose a heuristic estimation which method we term pilot-aware subspace
projection (PASP). The basic idea of this heuristic is to take the SVD of the
uplink data $\Yul=\vU\vS\vV\he$ and for each user $(i,k)$ select a subset
of the left singular vectors in $\vU$ based on the slow-fading coefficients
in $\vB$ and the allocation of pilot sequences. The least squares estimate
of the user $\vh_{ik}^\text{LS}$ is then projected onto the subspace
spanned by the selected singular vectors, i.e, 
\begin{align}
	\vh_{ik}^\text{PASP} = \vU \vLambda \vU\he \vh_{ik}^\text{LS}
\end{align}
where $\vLambda$ is a diagonal matrix with zeros and ones on the diagonal,
selecting the desired singular vectors.

A detailed description of the method and the rationale behind it can be
found in Appendix~\ref{sec:subspace} and results for the convergence speed in
simulations with different initializations are presented and discussed in Section~\ref{sec:results}.

Note that similarly to the blind MAP approach also for semi-blind channel
estimation the performance increases with an increasing number of antennas
and uplink data signals and thus the semi-blind approach is especially
suitable for the considered large-scale massive MIMO systems. However, in contrast
to the blind method, the semi-blind estimation offers a strict improvement
over the training based least-squares estimation and thus does not fall off
for a smaller number of available uplink signals.

\subsection{Upper Bound}
\label{sec:upper}

By combining the training signals with the uplink data signals in the
estimation process, we basically extend the training phase by the data phase
for which we only have statistical information. Therefore, in the hypothetical best
case, a semi-blind estimation method has the same performance as an optimal
estimation with exact knowledge of all transmitted data symbols. That is,
if a genie provides the transmitted data symbols, the augmented pilot
sequences in cell $i$ are given by
\begin{equation}
	\vPsi_i^\text{ub} = \bmat \vPsi_i \\ \vX_i \emat.
\end{equation}
We calculate the upper bound assuming known data symbols based on the
MAP/MMSE estimate discussed in Section~\ref{sec:training-based}.

\section{Results}
\label{sec:results}

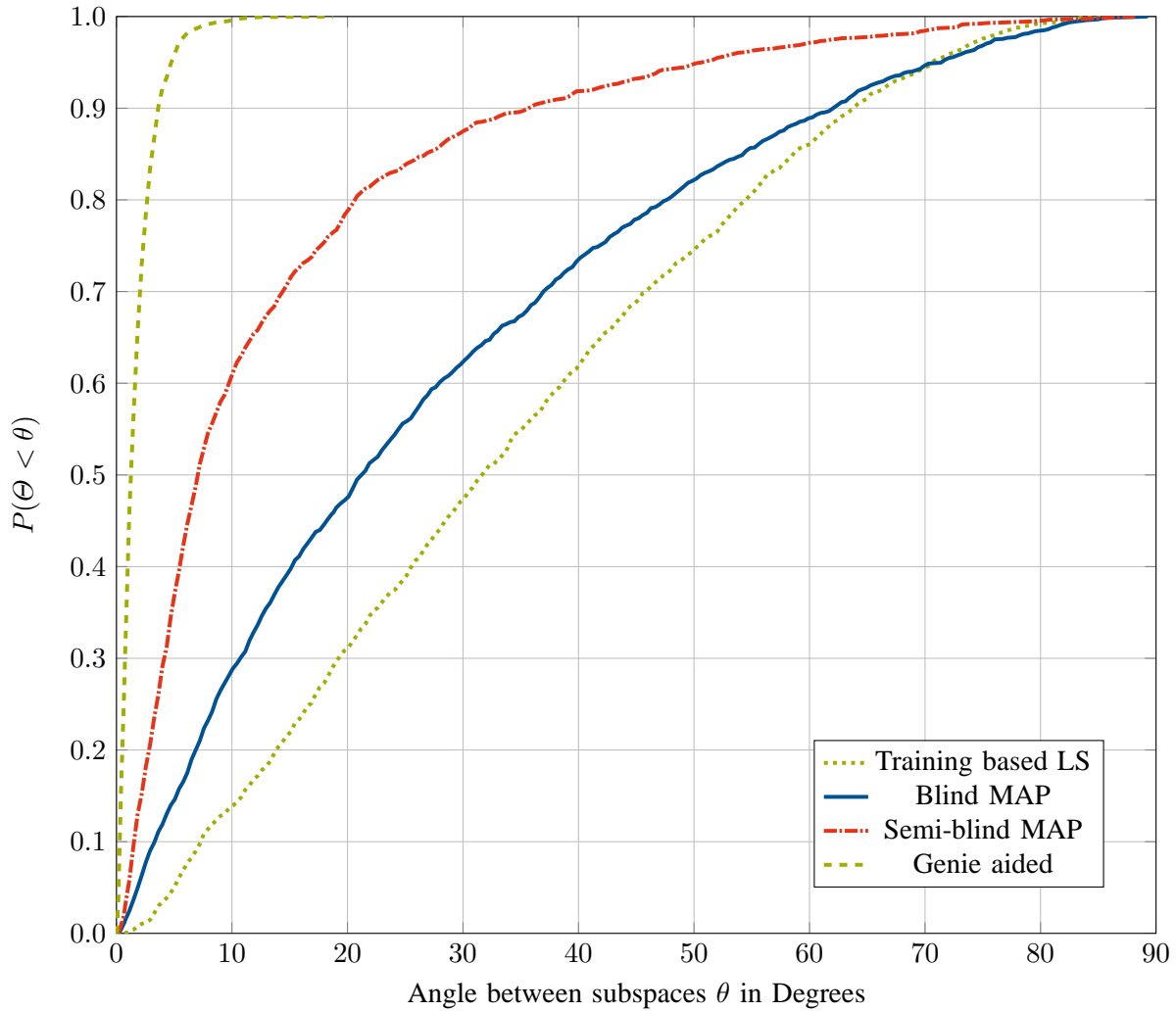
\begin{figure}[t]
	\begin{center}
\begin{tikzpicture}
\begin{axis}[
    height=0.85\columnwidth,
    width=0.95\columnwidth,
    grid=both,
    ticks=both,
    ymin=0,
    ymax=1,
    xmin=0,
    xmax=90,
    xlabel={Angle between subspaces $\theta$ in Degrees},
    ylabel={$P(\Theta < \theta)$},
    title={},
    yticklabel style={/pgf/number format/.cd, fixed, fixed zerofill, precision=1},
    legend entries={Training based LS,Blind MAP, Semi-blind MAP, Genie aided},
    legend style={at={(0.95,0.05)}, anchor=south east},
]

\addplot[ color=TUMGreen, dotted, line width=1.5pt ]
    table [x index=0,y index=1,col sep=comma] {err.csv};
\addplot[ color=TUMDarkerBlue, line width=1.5pt ]
    table [x index=0,y index=3,col sep=comma] {err.csv};
\addplot[ color=TUMBeamerRed, dash pattern=on 5pt off 1pt on 1pt off 1pt, line width=1.5pt ]
    table [x index=0,y index=5,col sep=comma] {err.csv};
\addplot[ color=TUMGreen, dashed,line width=1.5pt]
    table [x index=0,y index=6,col sep=comma] {err.csv};
\end{axis}
\end{tikzpicture} 
	\caption{CDF of the angle between estimated and actual subspace of the channel vectors for different estimation methods
		 in a network of $L=21$ cells with $K=4$ users per cell, $M=200$ base station antennas,
		and $\Tul=200$ uplink data samples.}
	\label{fig:errplot}
\end{center}
\end{figure}

\begin{figure}[h]
	\begin{center}
\begin{tikzpicture}
\begin{axis}[
    height=0.85\columnwidth,
    width=0.95\columnwidth,
    grid=both,
    ticks=both,
    ymin=0,
    ymax=1,
    xmin=0,
    xmax=11,
    xlabel={Matched filter rate $r_\text{MF}$},
    ylabel={$P(R_\text{MF}<r_\text{MF})$},
    title={},
    yticklabel style={/pgf/number format/.cd, fixed, fixed zerofill, precision=1},
    legend entries={Training based LS,Blind MAP, Semi-blind MAP, Genie aided},
    legend style={at={(0.95,0.05)}, anchor=south east},
]

\addplot[ color=TUMGreen, dotted, line width=1.5pt ]
    table [x index=0,y index=1,col sep=comma] {cdf_mf.csv};
\addplot[ color=TUMDarkerBlue, line width=1.5pt ]
    table [x index=0,y index=3,col sep=comma] {cdf_mf.csv};
\addplot[ color=TUMBeamerRed, dash pattern=on 5pt off 1pt on 1pt off 1pt, line width=1.5pt ]
    table [x index=0,y index=5,col sep=comma] {cdf_mf.csv};
\addplot[ color=TUMGreen, dashed,line width=1.5pt ]
    table [x index=0,y index=6,col sep=comma] {cdf_mf.csv};
\end{axis}
\end{tikzpicture} 
\begin{tikzpicture}
\begin{axis}[
    height=0.85\columnwidth,
    width=0.95\columnwidth,
    grid=both,
    ticks=both,
    ymin=0,
    ymax=1,
    xmin=0,
    xmax=11,
    xlabel={Zero-forcing rate $r_\text{ZF}$},
    ylabel={$P(R_\text{ZF}<r_\text{ZF})$},
    title={},
    yticklabel style={/pgf/number format/.cd, fixed, fixed zerofill, precision=1},
    legend entries={Training based LS,Blind MAP, Semi-blind MAP, Genie aided},
    legend style={at={(0.95,0.05)}, anchor=south east},
]

\addplot[ color=TUMGreen, dotted, line width=1.5pt ]
    table [x index=0,y index=1,col sep=comma] {cdf_zf.csv};
\addplot[ color=TUMDarkerBlue, line width=1.5pt ]
    table [x index=0,y index=3,col sep=comma] {cdf_zf.csv};
\addplot[ color=TUMBeamerRed, dash pattern=on 5pt off 1pt on 1pt off 1pt, line width=1.5pt ]
    table [x index=0,y index=5,col sep=comma] {cdf_zf.csv};
\addplot[ color=TUMGreen, dashed,line width=1.5pt ]
    table [x index=0,y index=6,col sep=comma] {cdf_zf.csv};
\end{axis}
\end{tikzpicture} 
	\caption{Empirical CDFs of average user rates for different estimation methods
		in a network of $L=21$ cells with $K=4$ users per cell,
		$M=200$ base station antennas and $\Tul=200$ uplink data
		samples. The upper plot depicts the matched filter results
		and the lower one the results for zero-forcing.}
	\label{fig:CDF}
\end{center}
\end{figure}

We compare the different approaches to channel estimation by simulation in
a cellular system with $L = 21$ cells in a wrap-around configuration, where
each base station is equiped with $M = 200$ antennas and $K=4$ users are
served in each cell. The pathloss of the users is calculated according to
the urban macro model in the ITU guidelines~\cite{itu-r_guidelines_2009} with log-normal shadow
fading with a standard deviation of $6$dB.

To evaluate the performance of the different algorithms we present results
for the angle between the estimated and the actual channel subspace, i.e.,
\begin{equation}\label{eq:subspace_angle}
	\theta(\vh,\hat\vh) = \cos\inv\left(\frac{\lvert \vh\he\hat\vh \rvert}{\lVert \vh \rVert \lVert \hat\vh \rVert}\right)
\end{equation}
and since we are evaluating a communication system, we also present results
for the achievable downlink rates, when using either a linear matched filter
or a linear zero-forcing precoder based on the channel estimates. For the
downlink rates we assume perfect knowledge of the equivalent channel at the
user terminals.

In Fig.~\ref{fig:errplot}, the experimental cummulative distribution
functions (CDF) of the subspace estimation error are depicted and in
Fig.~\ref{fig:CDF}, the corresponding downlink rates are shown. We give the
results for training-based least-squares estimation, blind estimation, and
semi-blind MAP estimation as well as the CDFs of the genie-aided upper
bound discussed in Section~\ref{sec:upper}.

Regarding the accuracy of the channel estimation measured by the angle
$\theta$ [see~\eqref{eq:subspace_angle}], the results in
Fig.~\ref{fig:errplot} show an improvement of the semi-blind method
compared to the LS channel estimation and the genie-aided channel
estimation gives a performance upper bound. However, the blind approach
outperforms the LS estimation in terms of accuracy of the estimated
subspace but delivers mixed results for the achievable rates with a slight
improvement for the worst 20\% of the users. The semi-blind MAP approach of
Section~\ref{sec:semi-blind} with the with pilot-aware subspace
projection~\eqref{eq:projection_pilot} as initialization always
yields a significant performance gain. In the considered scenario, the
average rate is increased by about 25\% when we apply semi-blind channel
estimation and especially the weak users benefit from the improved channel
estimation with gains at the 5th percentile of about 800\%. Also note that
zero-forcing precoding only benefits the users with above average channel
quality and for matched filter precoding, the blind estimation outperforms the
other methods for strong users due to the inherent orthogonality of the
channel estimates.

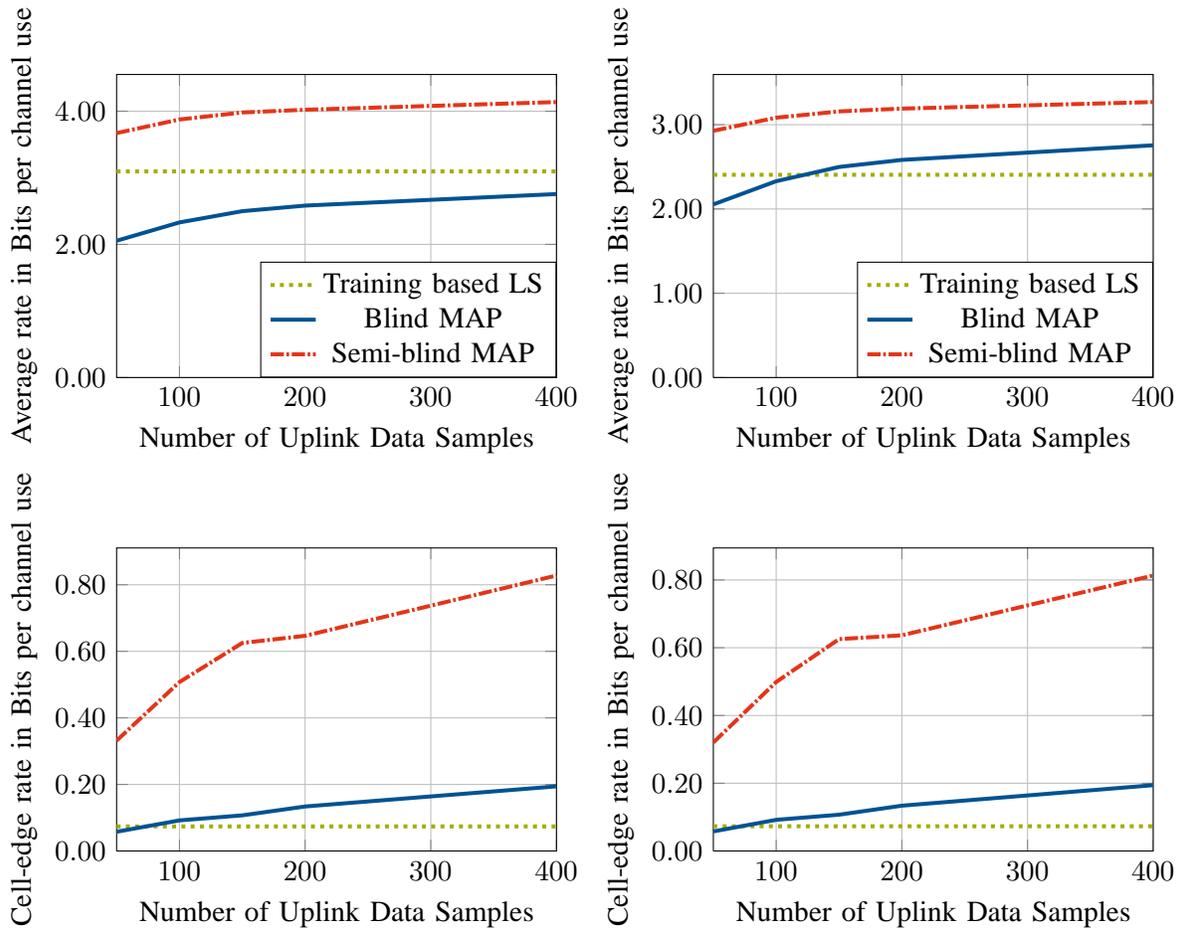
\begin{figure*}[h]
	\begin{center}
\begin{tikzpicture}
\begin{axis}[
    height=0.34\textwidth,
    width=0.45\textwidth,
    grid=both,
    ticks=both,
    xmin=50,
    xmax=400,
    ymin=0,
    xlabel={Number of Uplink Data Samples},
    ylabel={Average rate in Bits per channel use},
    title={},
    yticklabel style={/pgf/number format/.cd, fixed, fixed zerofill, precision=2},
    legend entries={Training based LS,Blind MAP, Semi-blind MAP},
    legend style={at={(1.00,0.0)}, anchor=south east},
%    legend pos=outer north east,
%    legend cell align=right,
]

\addplot[ color=TUMGreen, dotted, line width=1.5pt ]
    table [x index=0,y index=1,col sep=comma] {tul_zf.csv};
\addplot[ color=TUMDarkerBlue, line width=1.5pt ]
    table [x index=0,y index=3,col sep=comma] {tul_zf.csv};
\addplot[ color=TUMBeamerRed, dash pattern=on 5pt off 1pt on 1pt off 1pt, line width=1.5pt ]
    table [x index=0,y index=5,col sep=comma] {tul_zf.csv};

\end{axis}
\end{tikzpicture} 
\begin{tikzpicture}
\begin{axis}[
    height=0.34\textwidth,
    width=0.45\textwidth,
    grid=both,
    ticks=both,
    xmin=50,
    xmax=400,
    ymin=0,
    xlabel={Number of Uplink Data Samples},
    ylabel={Average rate in Bits per channel use},
    title={},
    yticklabel style={/pgf/number format/.cd, fixed, fixed zerofill, precision=2},
    legend entries={Training based LS,Blind MAP, Semi-blind MAP},
    legend style={at={(1.00,0.0)}, anchor=south east},
%    legend pos=outer north east,
%    legend cell align=right,
]

\addplot[ color=TUMGreen, dotted, line width=1.5pt ]
    table [x index=0,y index=1,col sep=comma] {tul_mf.csv};
\addplot[ color=TUMDarkerBlue, line width=1.5pt ]
    table [x index=0,y index=3,col sep=comma] {tul_mf.csv};
\addplot[ color=TUMBeamerRed, dash pattern=on 5pt off 1pt on 1pt off 1pt, line width=1.5pt ]
    table [x index=0,y index=5,col sep=comma] {tul_mf.csv};

\end{axis}
\end{tikzpicture} 

\begin{tikzpicture}
\begin{axis}[
    height=0.34\textwidth,
    width=0.45\textwidth,
    grid=both,
    ticks=both,
    xmin=50,
    xmax=400,
    ymin=0,
    xlabel={Number of Uplink Data Samples},
    ylabel={Cell-edge rate in Bits per channel use},
    title={},
    yticklabel style={/pgf/number format/.cd, fixed, fixed zerofill, precision=2},
    legend style={at={(1.00,0.0)}, anchor=south east},
%    legend pos=outer north east,
%    legend cell align=right,
]

\addplot[ color=TUMGreen, dotted, line width=1.5pt ]
    table [x index=0,y index=1,col sep=comma] {tul_zf_p5.csv};
\addplot[ color=TUMDarkerBlue, line width=1.5pt ]
    table [x index=0,y index=3,col sep=comma] {tul_zf_p5.csv};
\addplot[ color=TUMBeamerRed, dash pattern=on 5pt off 1pt on 1pt off 1pt, line width=1.5pt ]
    table [x index=0,y index=5,col sep=comma] {tul_zf_p5.csv};

\end{axis}
\end{tikzpicture} 
\begin{tikzpicture}
\begin{axis}[
    height=0.34\textwidth,
    width=0.45\textwidth,
    grid=both,
    ticks=both,
    xmin=50,
    xmax=400,
    ymin=0,
    xlabel={Number of Uplink Data Samples},
    ylabel={ Cell-edge rate in Bits per channel use},
    title={},
    yticklabel style={/pgf/number format/.cd, fixed, fixed zerofill, precision=2},
    legend style={at={(1.00,0.0)}, anchor=south east},
%    legend pos=outer north east,
%    legend cell align=right,
]

\addplot[ color=TUMGreen, dotted, line width=1.5pt ]
    table [x index=0,y index=1,col sep=comma] {tul_mf_p5.csv};
\addplot[ color=TUMDarkerBlue, line width=1.5pt ]
    table [x index=0,y index=3,col sep=comma] {tul_mf_p5.csv};
\addplot[ color=TUMBeamerRed, dash pattern=on 5pt off 1pt on 1pt off 1pt, line width=1.5pt ]
    table [x index=0,y index=5,col sep=comma] {tul_mf_p5.csv};

\end{axis}
\end{tikzpicture} 
\end{center}
	\caption{Achievable rate for different lengths of the uplink data
		transmission interval $\Tul$, on the left with zero-forcing
		and on the right with matched filter. The upper plots
		present the average rate and the lower ones the rate at the
		5th percentile, i.e., the rate of the cell-edge users. } 
	\label{fig:tul}
\end{figure*}

In Fig.~\ref{fig:tul}, we present the average rate performance of the different
algorithms versus the length of the received uplink data signals $\Tul$. For
the cell-edge users, the gain of the semi-blind MAP method is already
larger than 400\% for 50 uplink samples, growing to a tenfold increase in
performance for 400 uplink samples. The performance of blind estimation
increases only slowly with the number of uplink time slots.
In other words, in contrast to the blind MAP estimation, the proposed
semi-blind MAP estimation always outperforms the training based
least-squares estimate and offers significant performance gains even for a
small number of received uplink data signals. Note that for the blind and
the semi-blind estimation it is beneficial to have a large number of
antennas due to the increasing structure in the uplink data.

\begin{figure}[t]
	\begin{center}
\begin{tikzpicture}
\begin{axis}[
    height=0.85\columnwidth,
    width=0.95\columnwidth,
    grid=both,
    ticks=both,
    ymin=0,
    ymax=5,
    xmin=1,
    xmax=1280,
    xlabel={Number of L-BFGS Iterations},
    ylabel={Average user rate in Bit/s/Hz},
    title={},
    yticklabel style={/pgf/number format/.cd, fixed, fixed zerofill, precision=2},
    legend entries={ Pilot-aware Subspace Projection, Blind MAP, training based LS,random initialization},
    legend style={at={(1.00,0.0)}, anchor=south east},
%    legend pos=outer north east,
%    legend cell align=right,
]

\addplot[ color=TUMDarkerBlue, dotted,line width=1.5pt ]
    table [x index=0,y index=4,col sep=comma] {converge.csv};
\addplot[ color=TUMGreen, dash pattern=on 5pt off 2pt,line width=1.5pt ]
    table [x index=0,y index=3,col sep=comma] {converge.csv};
\addplot[ color=TUMOrange, dash pattern=on 5pt off 1pt on 1pt off 1pt, line width=1.5pt ]
    table [x index=0,y index=2,col sep=comma] {converge.csv};
\addplot[ color=black, line width=2pt ]
    table [x index=0,y index=1,col sep=comma] {converge.csv};

\end{axis}
\end{tikzpicture} 
\end{center}
	\caption{Performance for different initializations of the
		semi-blind MAP estimator with respect to the number of
		iterations used to find an optimizer for the MAP problem}
	\label{fig:init}
\end{figure}
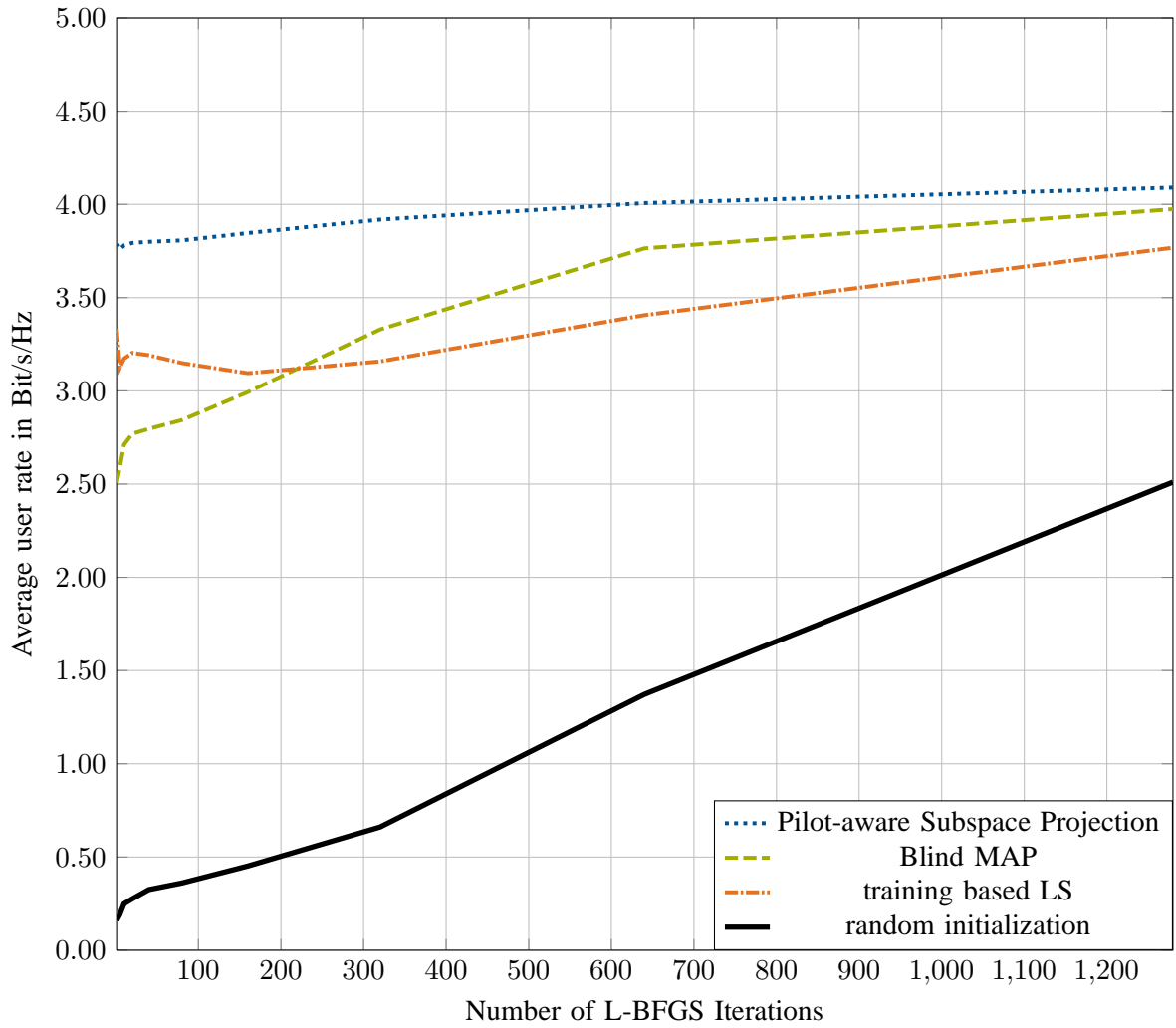

The unconstrained optimization for the MAP approach is performed with the
limited-memory BFGS (L-BFGS) algorithm~\cite{nocedal_updating_1980}. In
Fig.~\ref{fig:init}, the average user-rate performance of the MAP estimator is given for
different initializations versus the number of iterations of the
L-BFGS algorithm. Note that with all initializations, the iteration seems to converge
to the same achievable rate. However, the convergence is rather slow which is
due to the large dimensionality of the problem. 

If we use the least squares estimate as initial guess, the initial
improvement is small due to the fact that the estimates for the different
cells are linearly dependent or even identical when the same pilot
sequences are reused in each cell. Using the blind estimate as a starting
point yields improved convergence speed. The proposed semi-blind PASP
approach, however, significantly outperforms both the least-squares and the
blind estimates. Random initialization of the iteration does not result in
satisfactory performance for a reasonable amount of iterations and should
be avoided.

\section{Conclusion}

We discussed the MAP channel estimation based on different observations of
a base station in a cellular network. For training-based estimation the
results are well-known and involve solving a linear system of equations. We
were able to derive an quasi-analytical solution to the blind MAP problem based
on the SVD of the uplink data.
However, system level simulations indicate that blind estimation
delivers unsatisfactory results in a typical massive MIMO system. With the proposed
semi-blind MAP estimator on the other hand, we are able to improve upon the
estimation accuracy of a solely training based estimation. Our results
indicate that the proposed semi-blind estimation approaches consistently
outperform state-of-the-art training based channel estimation methods significantly.

\appendix

\subsection{Proof for the Blind MAP Estimator}
\label{sec:proof}

For the following derivation, we assume that $M > KL$, i.e., the number of
antennas at the considered base station is larger than the total number of
users. Additionally,
we need the probability density function of the uplink
data given the channel coefficients. Since we do not distinguish between
desired and interfering channels, we can write~\eqref{eq:sys_uplink} more
concisely as
\begin{equation}
	\vY = \sqrt{\pul} \vH \vX\he + \vN_\text{ul}
\end{equation}
where $\vX = [\vX_1,\ldots,\vX_L]$. Since all entries of both $\vX$ and
$\vN_\text{ul}$ are i.i.d. Gaussian with zero mean and unit variance, the
columns $\vy_t, t=1,\ldots,\Tul$, of the uplink data $\vY$, given the channel matrix $\vH$,
are also i.i.d. Gaussian with zero mean. The covariance matrix of the $t$-th column calculates to
\begin{equation}
	\expec[\vy_t\vy_t\he] = \pul \expec[\vH \vx_t \vx_t\he \vH\he] + \id = \pul \vH\vH\he + \id
\end{equation}
where $\vx_t\he$ is the $t$-th row of $\vX$.
The joint density for the uplink data $\vY$ given the channel realizations
$\vH$ thus reads as
\begin{align}
	f_{\Yul \vert \vect H}(\Yul \vert \vect H\vect) 
	&\propto  \frac{\exp\left[ -\tr\left[ \Yul\he \left(\pul\H \H\he + \id\right)^{-1} \Yul\right]\right]}
	{\det^{\Tul} \left(\pul \vect H^{\He} \vect H + \id\right)} \notag \\
&\propto  \frac{\exp\left[\tr\left[\Yul \Yul\he \vect H \left(\vect H^{\He} \vect H + 
\frac{1}{\pul}\id\right)^{-1} \vect H^{\He}\right]\right]}
	{\det^{\Tul} \left(\pul \vect H^{\He} \vect H + \id\right)}.
\end{align}
where we used the matrix inversion lemma to obtain the second line.

Due to the strict monotonicity of the logarithm, the MAP
formulation~\eqref{eq:MAPblindorig} can be rewritten as
\begin{equation}
	\H^\text{BL} = \argmax_{\vect H} l_\text{ul}(\vect H) + l_\text{pr}(\vect H)
	\label{eq:MAPblind}
\end{equation}
where 
\begin{equation}\label{eq:lh_ul}
l_\text{ul}(\vect H) = 
\left[\tr\left[\Yul \Yul\he \vect H \left(\vect H^{\He} \vect H + 
\frac{1}{\pul}\id\right)^{-1} \vect H^{\He}\right]\right] -\\
	\Tul \log\det \left(\pul \vect H^{\He} \vect H + \id\right)
\end{equation}
and [cf. \eqref{eq:lh_pr}]
\begin{equation}\label{eq:lh_pr_blind}
	l_\text{pr}(\vH) = -\tr[\H\B\inv\H\he]
\end{equation}
due to the channel model in~\eqref{eq:channel}.

In the following, we analyze the likelihood functions more closely. Let 
\begin{equation}\label{eq:svd}
	\H = \vW \vXi \vT\he \in \mathbb C^{M\times KL}
\end{equation}
denote the reduced singular value decomposition of $\H$ where
$\vXi\in\mathbb R^{Kl\times KL}$ is diagonal, $\vT \in \mathbb C^{KL \times
KL}$ is unitary, and $\vW \in \mathbb C^{M\times KL}$ is
subunitary due to the assumption that $M>KL$. Substituting the SVD of $\H$ into $l_\text{ul}$ yields
\begin{equation}
	l_\text{ul}(\vW, \vXi, \vT) = \tr\left[\Yul \Yul\he \vW \vXi \left(\vXi^2 + \frac{1}{\pul}\id\right)^{-1} \vXi\vW\he\right]\\
	- \log \det \left(\pul \vXi\vXi + \id\right).
\end{equation}
We observe that $l_\text{ul}$ does not depend on the right singular vectors
$\vT$ of $\H$, whereas
\begin{align}\label{eq:gpr}
	l_\text{pr}(\vW,\vXi,\vT) &= -\tr[\vW\vXi\vT\he \vB\inv \vT \vXi \vW\he] \notag \\
		&= -\tr[\vXi\vXi\vT\he\vB\inv\vT]
\end{align}
is independent of $\vW$.

For given singular values $\vXi$ and with
\begin{align}
\vC &= \Yul\Yul\he \\ 
\vD &= (\vXi^2 + \frac{1}{\pul}\id)\inv\vXi^2 \label{eq:diagd}
\end{align}
we can simplify $l_\text{ul}$ as
\begin{equation}
	l_\text{ul}(\vW) = \tr\left[\vW\he\vC \vW \vD \right]
\end{equation}
where the second term, which is independent of $\vW$, has been dropped.

As $l_\text{pr}(\vH)$ is independent of $\vW$, the optimal $\vW$ is the solution to
\begin{equation}\label{eq:leftsvopt}
	\vW^\text{opt} = \argmax_{\vW} \tr\left[\vW\he\vC \vW \vD \right] \st \vW\he \vW = \id.
\end{equation}
The corresponding Lagrangian is given by
\begin{equation}
	L(\vW,\vLambda) = \tr(\vW\he\vC \vW \vD) + \tr(\vLambda(\vW\he \vW - \id))
\end{equation}
where $\vLambda$ is the Hermitian Lagrangian multiplier. Derivation with respect to $\vW^*$ leads to
\begin{equation}
	\frac{\partial L(\vW,\vLambda)}{\partial \vW^*} = \vC \vW \vD + \vW \vLambda = \zeros
\end{equation}
from which follows that
\begin{equation}
	\vW\he \vC\vW\vD + \vLambda = \zeros
\end{equation}
since $\vW\he\vW = \id$. Consequently,
\begin{equation}
	\vW\he \vC\vW\vD = \vD\vW\he \vC\vW 
\end{equation}
since $\vLambda$ has to be Hermitian as the constraint
in~\eqref{eq:leftsvopt} is Hermitian. It can be inferred, that the optimal
left singular vectors $\vW$ of $\H$ diagonalize the sample covariance
$\vC$, i.e., with the eigenvalue decomposition $\vC = \vU \vSigm^2 \vU\he$
we have
\begin{equation}\label{eq:leftsv}
	\vW = \vU \vPi^\prime
\end{equation}
where $\vU$ is the matrix of left singular vectors of $\vY$ and $\vPi^\prime \in
\{0,1\}^{M\times KL}$ is the left block of a $M\times M$ permutation
matrix. Note that there is an ambiguity in every eigenvector of $\vC$,
i.e., the corresponding column of $\vU$, with respect to a scalar
multiplication with absolute value one, which we can move to the right
singular vectors $\vT$ of $\H$ so we do not have to consider it here.

Substituting \eqref{eq:leftsv} into~\eqref{eq:leftsvopt} yields
\begin{equation}
	{\vPi}^\text{opt} = \argmax_{\vPi^\prime} \tr\left[(\vPi^\prime)\tp \vect \Sigma^2 \vPi^\prime \vD \right]
\end{equation}
where $\vPi^\prime \in \{0,1\}^{M\times KL}$.
Since the diagonal entries in $\vD$ are ordered decreasingly, as the
entries of $\vXi$ are in descending order~[cf~\eqref{eq:diagd}], the
optimal choice for the selection matrix $\vPi^\prime$ is 
\begin{equation}
	\vPi^\text{opt} = \bmat \id \\ \zeros \emat
\end{equation}
because we want to match the largest eigenvalues in $\vect \Sigma^2$
with the largest values in $\vD$ to maximize the trace. Thus, the optimal choice for the left
singular vectors $\vW$ of $\H$, when conducting blind estimation, are simply the
principal left singular vectors $\vU_{1:KL}$ of $\Yul$.

A similar analysis can applied to $l_\text{pr}(\vH)$ [see~\eqref{eq:gpr}],
which is independent of $\vW$, to get the optimal right singular vectors
$\vT$. For given $\vXi$, the optimal right singular vectors are given by
\begin{equation}\label{eq:right_sv}
	\vT = \vPi\vPhi
\end{equation}
where $\vPi \in \{0,1\}^{KL \times KL}$ is a permutation matrix, such that the entries along the
diagonal of $\vPi\tp \vB\inv \vPi$ are sorted ascendingly and $\vPhi$ contains
the unknown phase shifts mentioned above. Based on the permutation matrix
$\vPi$, we can define a function $\pi(i,k)$ such that the $\pi(i,k)$th
left singular vector $\vu_{\pi(i,k)}$ of $\Yul$ spans the subspace of the blind
estimate $\vh^\text{BL}_{ik}$ of user $k$ in cell $i$. The inverse mapping $(i,k) =
\pi\inv(n)$ delivers the cell and user indices corresponding to the $n$th singular vector
of the uplink data.

With the results for the left and right singular vectors, we can now
optimize the singular values $\vXi$. Substituting the previous results, the
optimization problem reads as
\begin{align}
	\vXi^\text{opt} &= \argmax_{\vXi \succeq \zeros } 
	\begin{array}[t]{l} 
		\tr\left[\vect \Sigma_p^2 \vXi^2 \left(\vXi^2 + \frac{1}{\pul}\id\right)^{-1} \right] \\[2mm]
		- \Tul \log \det \left(\pul \vXi^2 + \id\right)\\[2mm]
		- \tr[ \vXi^2 \vPi\tp \vB\inv \vPi] 
	\end{array} \notag \\
	&= \argmax_{\vXi \succeq \zeros} l(\vXi^2)
\end{align}
where $\vSigm_\text{P} = \diag(\sigma_1, \ldots, \sigma_{KL})$ denotes the
matrix with the principal $KL$ singular values of $\Yul$ along the diagonal.

Setting the derivatives of the objective function
\begin{equation}
	\frac{\partial l(\vXi^2)}{\partial \xi_n^2} = \frac{\sigma_n^2/\pul}{(\xi_n^2 + 1/\pul)^2} - \frac{\Tul}{\xi_n^2 + 1/\pul} - \frac{1}{\beta_{\pi\inv(n)}}
\end{equation}
to zero yields
\begin{align}\textstyle{
	\frac{\beta_{\pi\inv(n)}\sigma_n^2}{\pul} - \beta_{\pi\inv(n)}\Tul\left(\xi_n^2 + \frac{1}{\pul}\right) - \left(\xi_n^2 + \frac{1}{\pul}\right)^2} &= 0\\
\textstyle{\xi_n^4 + \left(\frac{2}{\pul} + \beta_{\pi\inv(n)}\Tul\right) \xi_n^2 
	+ \frac{\beta_{\pi\inv(n)}(\Tul-\sigma_n^2)}{\pul} + \frac{1}{\pul^2} }  &= 0.
\end{align}
Since $\xi_n^2 \geq 0$, the optimal $\xi_n^2$ are thus given by
\begin{equation}
	\xi_n^2 = \left[-\frac{1}{\pul} - \frac{\beta_{\pi\inv(n)}\Tul}{2} + \sqrt{ \frac{\beta_{\pi\inv(n)}^2\Tul^2}{4} + \frac{\beta_{\pi\inv(n)}\sigma_n^2}{\pul} } \right]_+.
\end{equation}

\subsection{Pilot-Aware Subspace Projection}
\label{sec:subspace}

From Section~\ref{sec:blind}, we know that the blind estimation of the
channel is unique up to a scalar ambiguity. The direction information
resulting from the blind estimation of the channel to user $k$ in cell $i$ is
given by the $\pi(i,k)$th left singular vector $\vu_{\pi(i,k)}$ of the
uplink data $\vY$ [see the discussion below~\eqref{eq:right_sv}].

By projecting the training based least-squares estimate onto the
one-dimensional space spanned by the blind estimate, we get a unique
estimate of the channel vector
\begin{equation}\label{eq:projection_single}
	\vh_{ik}^\text{PR} = \vu_{\pi(i,k)} (\vu_{\pi(i,k)}\he \vh_{ik}^\text{LS}).
\end{equation}
If the direction of the blind estimate is accurate, the projection cancels
out most of the interference in the least-squares estimate. However, the
projection in~\eqref{eq:projection_single} does not improve the estimation
of the spatial direction of the channel vector w.r.t. the blind estimate.
Indeed, it only removes the scalar ambiguity of the blind estimate, which
is, however, of small importance or even irrelevant in most cases.

Alternatively, we can collect $R > K$ left singular vectors of $\Yul$
corresponding to the $R$ largest singular values in the matrix $\vU_{1:R}$
and project the least-squares estimate onto the range of $\vU_{1:R}$, i.e.,
\begin{equation}\label{eq:projection_full}
	\vh_{ik}^\text{SUB} = \vU_{1:R} \vU_{1:R}\he \vh_{ik}^\text{LS}.
\end{equation}
Note that the number of basis vectors $R$ must be chosen large enough such that all significant parts of the
desired channel vectors lie in the resulting subspace. With this heuristic strategy, 
noise and weak interferers can be suppressed. However, the channel estimate is still
subject to pilot-contamination originating from strong interfering users.

We can generalize on the semi-blind approaches
in~\eqref{eq:projection_single} and~\eqref{eq:projection_full} by incorporating a
diagonal weighting matrix into the projection resulting in
\begin{equation}\label{eq:projection_general}
	\vh_{ik}^\text{PASP} = \vU_{1:R} \diag(\vla_{ik}) \vU_{1:R}\he \vh_{ik}^\text{LS}
\end{equation}
where the weights $\vla_{ik}$ are chosen heuristically and depend on both,
the slow fading coefficients and the pilot sequences of all relevant users.

If we choose $\vla_{ik} = \ones$, we get the full projection on the $R$-dimensional subspace
in~\eqref{eq:projection_full} which is similar to the subspace projection
proposed in~\cite{mueller_blind_2013}, but extended to the practically relevant
case of strong interferers. For $\vla_{ik} = \ve_{\pi(i,k)}$, we get the
projection in~\eqref{eq:projection_single}, which performs equivalent to
the blind estimate.

Our goal is to find a heuristic, which combines the subspace information
from the blind estimation with the LS estimate obtained based on the training data in a
way such that the interference in the LS estimate is reduced. 
Suppose all base
stations use a fixed set of orthonormal training sequences $\set P$ with
$\lvert \set P \rvert = \Ttr$, i.e., we have $\vpsi_{ik} \in \set P$ for
each training sequence in all cells. Therefore, the channel estimate of a desired user suffers from
contamination originating from users in neighboring cells that employ the same training
sequence. 

By the choice of $\vla_{ik}$, the LS estimate is projected onto the space spanned by several of the
singular vectors, that is, onto $\spn(\{\vu_n\}_{n=a}^b)$ with $a\leq \pi(i,k) \leq b$, to
reduce the impact of inaccuracies in the singular vectors due to users with
similar channel gains, while still reducing significant parts of the
interference.

Let us denote with $l<\pi(i,k)$ the index of the closest (in
terms of channel quality) interfering user with better channel quality and
$u>\pi(i,k)$ the index of the closest interfering user with lower channel
quality. 
To achieve the goal of interference suppression, none of the singular vectors $\vu_n, n =
a,a+1,\ldots,b$ should correspond to an interfering user, e.g., 
 with the following heuristic choice for the window coefficients
\begin{equation}\label{eq:projection_pilot}
	\lambda_{ik,n} = \begin{cases}
		1 \quad &\text{for} \;\; \sqrt{\beta_{\pi\inv(l)}\beta_{ik}} \geq \beta_{\pi\inv(n)} \geq \sqrt{ \beta_{\pi\inv(u)} \beta_{ik}} \\
		0 \quad &\text{else}.
	\end{cases}
\end{equation}
In other words, we project onto the basis of the left singular vectors of
$\vY$ that are ``closer'' to that corresponding to the desired user than to
any interfering user, where we use the geometric mean to define the point
of equivalent distance.

With this approach, we project the training based estimate onto a higher
dimensional subspace which leads to an improved estimate of the channel
direction, while still suppressing major parts of the pilot-contamination.

Note that it is not possible to estimate the CSI accurately by this
heuristic method, if an interfering user, which employs the same pilot
sequence, has a slow fading coefficient that is very similar to that of the
desired user. This indicates, that the estimation performance can be
improved by assigning pilots accordingly.

\bibliographystyle{IEEEtran}
\bibliography{literature}

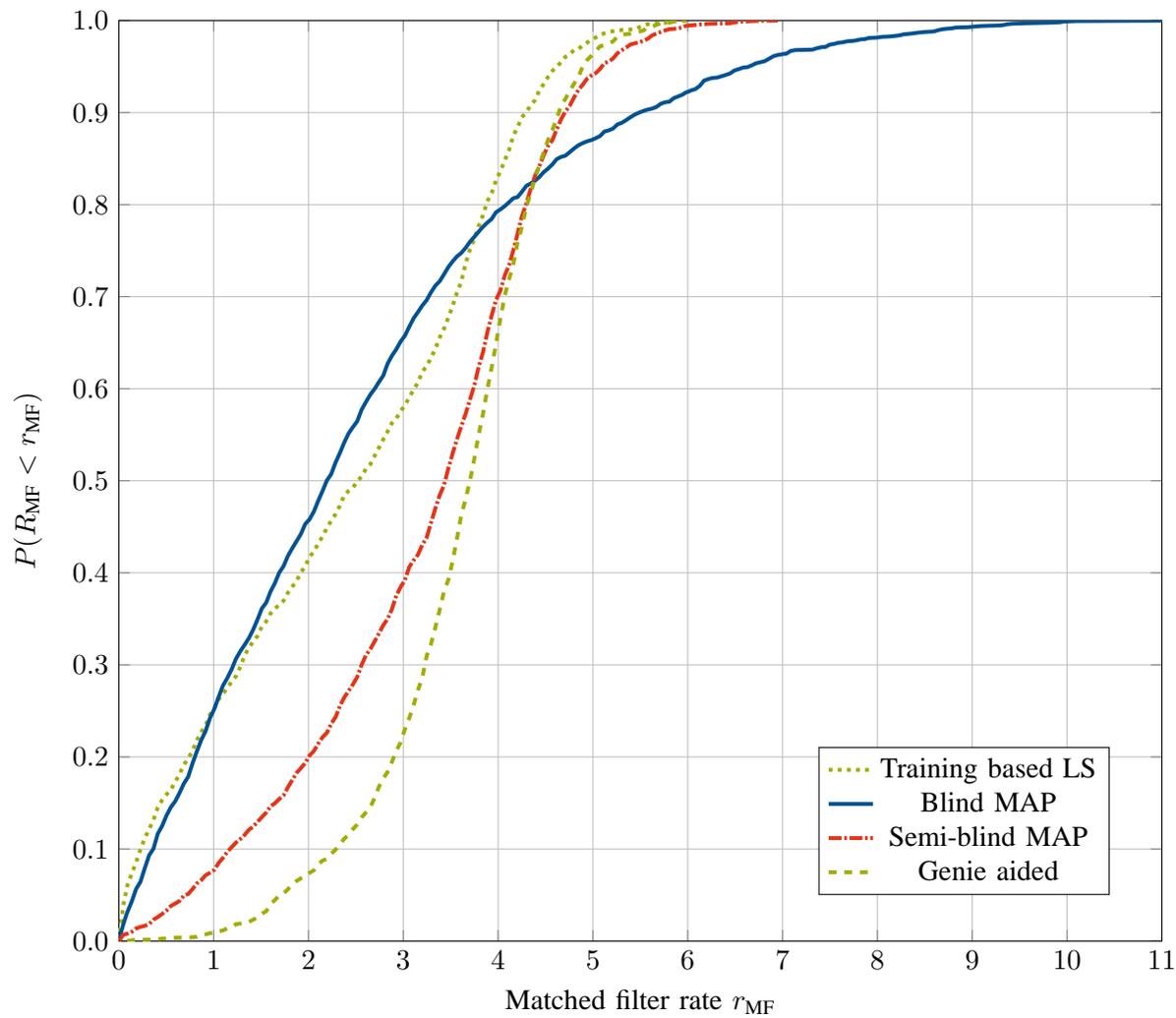
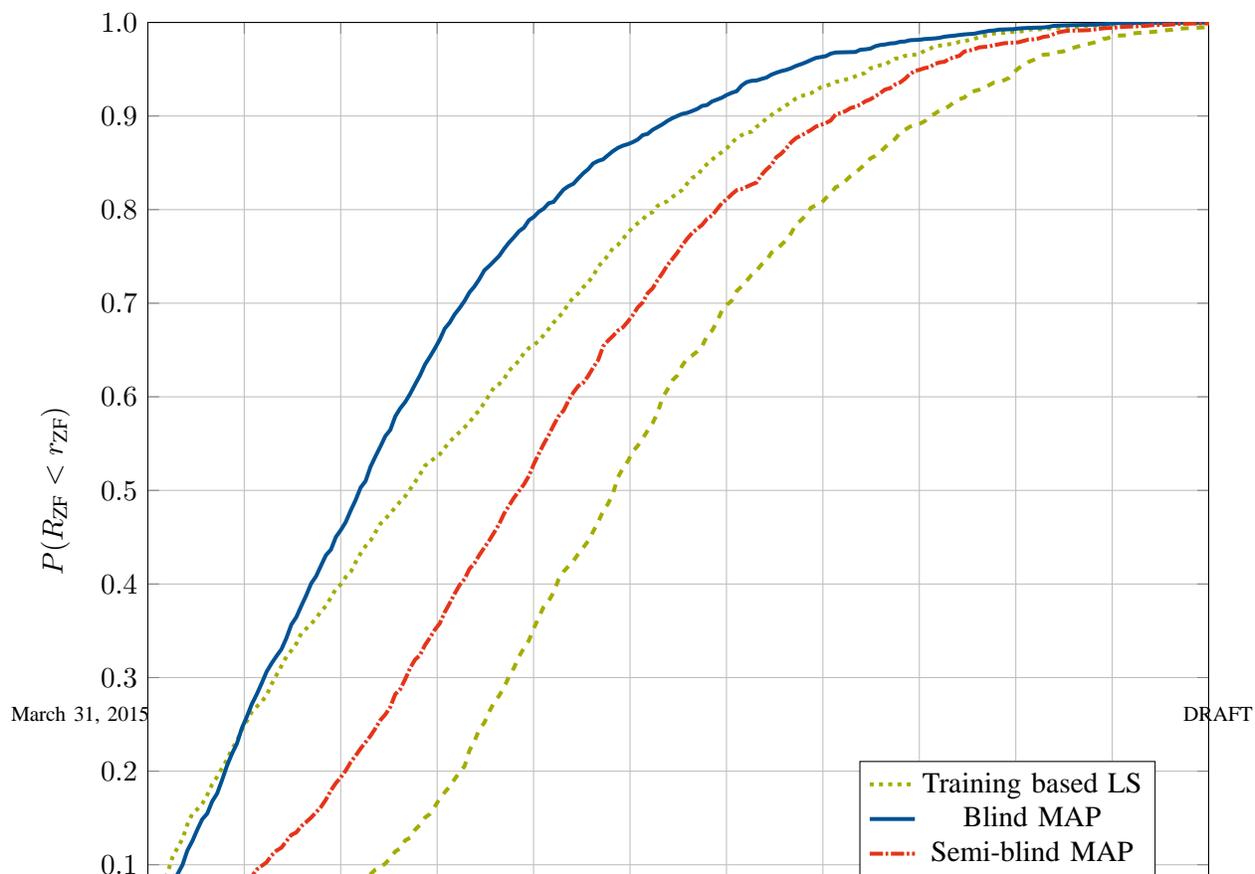
\end{document}